\documentclass[3p,twocolumn]{elsarticle}
\usepackage{lipsum}
\makeatletter
\def\ps@pprintTitle{%
 \let\@oddhead\@empty
 \let\@evenhead\@empty
 \def\@oddfoot{}%
 \let\@evenfoot\@oddfoot}
\makeatother
\usepackage{amsmath, amssymb, amsfonts, amsthm}
\usepackage{graphicx}
\usepackage{gastex}
\usepackage{wrapfig}
\theoremstyle{plain}
\usepackage{lineno}

\newtheorem*{theorem}{Theorem}

\newtheorem*{conjecture}{Conjecture}
\theoremstyle{definition}

\newcommand{\LCE}{\mathit{LCE}}

\journal{Information Processing Letters}

\begin{document}

\begin{frontmatter}

\title{Tight Lower Bounds for the Longest Common Extension Problem}

\author{Dmitry Kosolobov} 
\address{University of Helsinki, Helsinki, Finland}

\begin{abstract}
The longest common extension problem is to preprocess a given string of length $n$ into a data structure that uses $S(n)$ bits on top of the input and answers in $T(n)$ time the queries $\mathit{LCE}(i,j)$ computing the length of the longest string that occurs at both positions $i$ and $j$ in the input. We prove that the trade-off $S(n)T(n) = \Omega(n\log n)$ holds in the non-uniform cell-probe model provided that the input string is read-only, each letter occupies a separate memory cell, $S(n) = \Omega(n)$, and the size of the input alphabet is at least $2^{8\lceil S(n) / n\rceil}$. It is known that this trade-off is tight.
\end{abstract}

\end{frontmatter}


\section{Introduction}

Data structures for solving the so-called \emph{longest common extension (LCE)} problem (sometimes referred to as the longest common prefix problem) play the central role in the wide range of string algorithms. In this problem we must preprocess an input string of length $n$ so that one can answer the queries $\LCE(i,j)$ computing the length of the longest string that occurs at both positions $i$ and $j$ in the input. Since the existing solutions to this problem often, in practice particularly, constitute a bottleneck either in space or in time of the algorithms relying in their core on the LCE queries, many efforts have been made in the past decades to develop better LCE data structures.

In this paper we prove that the trade-off\footnote{For brevity, $\log$ denotes the logarithm with base~$2$.} $S(n)T(n) = \Omega(n\log n)$ holds for any data structure that solves the LCE problem using $S(n)$ bits of space (called \emph{additional space}) on top of the input and $T(n)$ time for the LCE queries, assuming that the input string is read-only, each letter occupies a separate memory cell, and $S(n) = \Omega(n)$ (such space is used in most applications of the LCE problem). For $S(n) = \Omega(n)$, this new trade-off improves by $\log n$ factor the trade-off $S(n)T(n) = \Omega(n)$ established by Bille et al.~\cite{BilleGortzSachVildojJ}, who used a simple reduction to a lower bound obtained by Brodal et al.~\cite{BrodalDavoodiRao} for the so-called range minimum queries problem.

Our result is proved in the \emph{cell-probe model}~\cite{Miltersen}, in which the computation is free and time is counted as the number of cells accessed (probed) by the query algorithm. The algorithm is also allowed to be non-uniform, i.e., we can have different algorithms for different sizes $n$ of the input. We assume that each letter of the input string is an integer located in a separate memory cell and each cell can store any integer from the set $\{0, 1, \ldots, n{-}1\}$. Hence, the maximal size of the input alphabet is $n$; this is a common assumption justified in, e.g., \cite{BurkhardtKarkkainen}. However, our main theorem poses a more specific restriction: the size of the input alphabet must be at least $2^{8\lceil S(n) / n\rceil}$. For instance, our trade-off is applicable for constant alphabets if $S(n) = \Theta(n)$, but to apply the trade-off in the case $S(n) = \Theta(n\sqrt{\log n})$, we have to have an alphabet of at least $2^{\Omega(\sqrt{\log n})}$ size.

\paragraph{Overview of LCE data structures}
The classical solutions for the LCE problem use $\Theta(n\log n)$ bits of space and $O(1)$ time for queries (e.g., see~\cite{FischerHeun,HarelTarjan}). In~\cite{BilleEtAl} Bille et al.~presented a RAM data structure that solves the LCE problem using $O(\tau)$ time for queries and $O(\frac{n\log n}{\tau})$ bits of additional space, where $\tau$ is a parameter such that $1 \le \tau \le n$. This result shows that our trade-off is tight and cannot be improved. The construction time of this data structure (in $O(\frac{n\log n}{\tau})$ bits of space) is $O(n^{2+\varepsilon})$, which is unacceptably slow. In~\cite{TIBIPT} Tanimura et al. proposed a data structure with significantly better $O(n\tau)$ construction time within the same $O(\frac{n\log n}{\tau})$ bits of additional space but with slightly suboptimal query time $O(\tau \min\{\log\tau, \log\frac{n}{\tau}\})$.

Denote by $\sigma$ the size of the input alphabet. Recently, Tanimura et al.~\cite{TNBIT} described a data structure that, for $\sigma \le 2^{o(\log n)}$, uses $o(n\log n)$ bits of additional space and $O(1)$ time for LCE queries thus ``surpassing'' our trade-off and showing the importance of the condition $\sigma \ge 2^{8\lceil S(n) / n\rceil}$. We believe also that our trade-off does not hold if the algorithm can read $\Omega(\log_{\sigma} n)$ consecutive letters of the input string in $O(1)$ time packing them in one $\Omega(\log n)$-bit machine word; this model reflects the situation that one can often observe in practice.

All mentioned results consider applications in which the input string is treated as read-only. In practice, however, we usually need a data structure that provides fast access to the string and allows us to answer the LCE queries, but the space occupied by the string itself can be reorganized. The data structure of~\cite{FischerIKopple} using this model occupies $O(\frac{n\log n}{\tau})$ bits of additional space and answers LCE queries in $O(\log^* n (\log\frac{n}{\tau} + \tau^{\log 3} / \log_{\sigma} n))$ time, where $\tau$ is a parameter such that $1 \le \tau \le n$ (however, this result still does not break our trade-off). The construction time for this structure (in $O(\frac{n\log n}{\tau})$ bits) is $O(n(\log^* n + \frac{\log n}{\tau} + \frac{\log \tau}{\log_{\sigma} n}))$. In~\cite{Prezza} Prezza described an ``in-place'' data structure\footnote{The data structure uses only negligible $O(\log^2 n)$ bits of space on top of the input.} that replaces $n\lceil\log\sigma\rceil$ bits occupied by the input with a data structure that allows to retrieve any substring of length $m$ of the input in optimal $O(\frac{m}{\log_{\sigma} n})$ time and answers the LCE queries in $O(\log\ell)$ time, where $\ell$ is the result of the query.\footnote{A similar result in~\cite{PolicritiPrezza} seems to be very practical, but its correctness in the RAM model, where $n$ tends to infinity, relies on a questionable assumption that the natural density of the logarithms of the Mersenne primes is non-zero (this is required to process these primes in constant time with $\Theta(\log n)$-bit machine words.)} For his data structure, Prezza presents only a randomized construction algorithm working in $O(n\log n)$ expected time and $O(n\log n)$ bits of space.

In certain applications the exact accuracy of the LCE queries is less important than construction time, query time, and space. For such applications, several Monte Carlo data structures were developed: their construction algorithm builds with high probability (i.e., with probability $1 - \frac{1}{n^c}$ for any specified constant $c > 0$) a valid data structure answering any LCE query correctly but sometimes can produce a faulty data structure. Prezza~\cite{Prezza} described a Monte Carlo version of his ``in-place'' data structure that answers the LCE queries in $O(\log\ell)$ time and has a construction algorithm working in $O(\frac{n}{\log_{\sigma} n})$ expected time using the same memory, i.e., also ``in-place''. Bille et al.~\cite{BilleEtAl} presented a~Monte Carlo version of their data structure for read-only inputs that answers the LCE queries in $O(\tau)$ time using $O(\frac{n\log n}{\tau})$ bits of additional space and has $O(n\log\frac{n}{\tau})$ construction time (within the same space), where $1 \le \tau \le n$. Gawry\-chowski and Kociumaka~\cite[Th. 3.3]{GawrychowskiKociumaka} described a modification of this Monte Carlo solution for read-only inputs that has the same optimal space and query time bounds but can be constructed in optimal $O(n)$ time.

Recently, several LCE data structures for compressed strings were developed. For a more detailed discussion on this topic, we refer the reader to~\cite{I,TNBIT} and references therein.

\section{Main Result}

\paragraph{Preliminaries}
A \emph{string $s$ of length $n$} over an alphabet $\Sigma$ is a map $\{0,1,\ldots,n{-}1\} \mapsto \Sigma$, where $n$ is referred to as the \emph{length of $s$}, denoted by $|s|$. We write $s[i]$ for the $i$th letter of $s$. A string $s[0]s[1]\cdots s[j]$ is a \emph{prefix} of $s$. For any $i$ and $j$, the set $\{k\in \mathbb{Z} \colon i \le k \le j\}$ (possibly empty) is denoted by $[i..j]$.

\begin{theorem}
In the non-uniform cell-probe model the trade-off $S(n) T(n) = \Omega(n\log n)$ holds for any data structure that solves the LCE problem for a read-only string of length $n$ using $S(n)$ bits of space and $T(n)$ time for queries assuming that each input letter occupies a separate cell, the size of the input alphabet is at least $2^{8\lceil S(n)/n\rceil}$, and $S(n) = \Omega(n)$.
\end{theorem}
\begin{proof}
Without loss of generality, assume that $T(n) \ge 1$. Suppose, for the sake of contradiction, that $S(n)T(n) \notin \Omega(n\log n)$. Then, there is an infinite set $N$ of positive integers such that $\lim_{n\in N} \frac{S(n)T(n)}{n\log n} = 0$. Hence, we obtain $\lim_{n\in N} \frac{S(n)}{n\log n} = 0$. Therefore, there is a positive function $\varepsilon(n)$ such that $S(n) = \varepsilon(n) n \log n$ for $n \in N$ and $\varepsilon(n)$ tends to $0$ as $n \to +\infty$.

Let us first construct a family $\mathcal{F}$ of inputs for the subsequent analysis. Define $\sigma = 2^{8\lceil \varepsilon(n)\log n\rceil}$. The input alphabet is $[1 .. \sigma]$. Note that $\sigma \ge 2^8$ for $n > 1$ and $\sigma = 2^{8\lceil S(n)/n\rceil}$ for $n \in N$. Since $\log\sigma = o(\log n)$ and, consequently, $\sigma < n$ for sufficiently large $n$, each letter of the alphabet fits in one memory cell. Observe that, since $S(n) = \varepsilon(n)n\log n \le \frac{1}8 n\log\sigma$ for $n \in N$, we are not able to encode the whole string in $S(n)$ bits and answer the LCE queries without any access to the string.

Denote $k = \lfloor\frac{1}2 \log_{\sigma} n\rfloor$. Since $\log\sigma =o(\log n)$, we have $k = \Theta(\frac{\log n}{\log\sigma}) = \omega(1)$. Therefore, $k \ge 1$ for sufficiently large $n$. Let $s_1, s_2, \ldots, s_{\sigma^k}$ denote all strings of length $k$ over the alphabet $[1..\sigma]$. The family $\mathcal{F}$ consists of all strings of the form $s_1 s_2 \cdots s_{\sigma^k} t$, where $t$ is a string of length $n - k\sigma^k$ over the alphabet $[1..\sigma]$. Since $\sigma^k k \le \sqrt{n}\log n$, it is easy to verify that $\frac{1}2 n \le n - k\sigma^k$ for $n \ge 2^8$. Hence, we obtain $|\mathcal{F}| \ge \sigma^{\frac{1}2n} = 2^{\frac{1}2 n\log\sigma}$ for $n \ge 2^8$. For convenience, we assume hereafter that $\min N \ge 2^8$.

By the pigeonhole principle, there is a subfamily $\mathcal{I} \subseteq \mathcal{F}$ such that $|\mathcal{I}| \ge |\mathcal{F}| / 2^{S(n)}$ and, for any strings $s, s' \in \mathcal{I}$, the encodings of $s$ and $s'$ in the $S(n)$ bits of the considered LCE data structure are equal. Since $S(n) \le \frac{1}8 n\log\sigma$ for $n \in N$, we obtain $|\mathcal{I}| \ge |\mathcal{F}| / 2^{S(n)} \ge 2^{\frac{1}2 n\log\sigma - S(n)} \ge 2^{\frac{1}2 n\log\sigma - \frac{1}8 n\log\sigma} = 2^{\frac{3}{8} n\log\sigma}$ for $n \in N$.

Let us prove that, for each $s \in \mathcal{I}$, there is a set of positions $T_s \subseteq [0..n{-}1]$ such that $|T_s| \le \frac{T(n)n}{k}$ and, for any string $s' \in \mathcal{I}$, we have $s = s'$ iff $s[i] = s'[i]$ for every $i \in T_s$. Choose $s \in \mathcal{I}$; $s = s_1 s_2 \cdots s_{\sigma^k} t$ for a string $t$. Without loss of generality, assume that $|t|$ is a multiple of $k$ (the case $|t| \bmod k \ne 0$ is similar). Since $\{s_1, s_2, \ldots, s_{\sigma^k}\}$ is the set of \emph{all} strings of length $k$ over the alphabet $[1..\sigma]$, there must exist exactly one sequence $i_1, i_2, \ldots, i_r$, where $r = |t| / k$, such that $t = s_{i_1} s_{i_2} \cdots s_{i_r}$ and $i_j \in [1..\sigma^k]$ for $j \in [1..r]$. Let $j \in [1..r]$. By the definition of $\LCE$, the query $\LCE(|s_1 s_2 \cdots s_{i_j-1}|, |s_1s_2 \cdots s_{\sigma^k} s_{i_1} s_{i_2} \cdots s_{i_{j-1}}|)$ reads at most $T(n)$ letters of the string $s$ and returns as an answer a number that is not less than $|s_{i_j}|$. Denote by $T^j_s$ the set of all positions $i \in [0..n{-}1]$ such that the letter $s[i]$ was accessed by the query. Define $T_s = T^1_s \cup T^2_s \cup \ldots \cup T^r_s$. It is easy to see that $|T_s| \le \frac{T(n) |t|}k \le \frac{T(n) n}{k}$. If a string  $s' \in \mathcal{I}$ coincides with the string $s$ on the positions $T_s$, then, by the definition of $T_s$, any query $\LCE(|s_1 s_2 \cdots s_{i_j-1}|, |s_1 s_2 \cdots s_{\sigma^k} s_{i_1} s_{i_2} \cdots s_{i_{j-1}}|)$ on the string $s'$ must return the same result as the corresponding query on the string $s$, i.e., the algorithm cannot distinguish $s$ and $s'$ on such queries. Thus, the numbers computed by these queries are not less than $|s_{i_j}|$. Consequently, since the algorithm is assumed to be correct and the string $s_1 s_2 \cdots s_{\sigma^k}$ is a common prefix of $s$ and $s'$, the string $s'$ must be equal to $s$.

Clearly, there are at most $2^n$ subsets $T_s$ (as there are at most $2^n$ subsets of $[0..n{-}1]$). Thus, by the pigeonhole principle, there exists a subfamily $\mathcal{I'} \subseteq \mathcal{I}$ such that $|\mathcal{I'}| \ge |\mathcal{I}| / 2^n$ and $T_s = T_{s'}$ whenever $s, s' \in \mathcal{I'}$. Since $|\mathcal{I}| \ge 2^{\frac{3}8 n\log\sigma}$ for $n \in N$ and $\log\sigma \ge 8$,  we obtain $|\mathcal{I'}| \ge |\mathcal{I}| / 2^n \ge 2^{\frac{3}{8} n\log\sigma - n} = 2^{\frac{1}{4} n\log\sigma + \frac{1}8 n\log\sigma - n} \ge 2^{\frac{1}{4} n\log\sigma}$ for $n \in N$.

By the choice of $\mathcal{I'}$, the set $T_s$ is the same for every $s\in \mathcal{I'}$. Denote this set by $T$. By the definition of $T$, the size of the family $\mathcal{I'}$ is upper bounded by the number $\sigma^{|T|} = 2^{|T|\log\sigma} \le 2^{T(n) n\log\sigma / k}$. If $T(n) < \frac{1}4 k$ for arbitrarily large numbers $n \in N$, then $\sigma^{|T|} < 2^{\frac{1}{4} n\log\sigma} \le |\mathcal{I'}|$, which is a contradiction. Thus, we obtain $T(n) \ge \frac{1}4 k$ for all sufficiently large $n \in N$ and hence $S(n)T(n) \ge \varepsilon(n) n \log n \cdot \frac{1}4 k$. Since $S(n) = \varepsilon(n) n \log n$ for $n \in N$ and $S(n) = \Omega(n)$, the function $\varepsilon(n)$ can be chosen so that $\varepsilon(n) \log n = \Omega(1)$. Hence, $k\, \varepsilon(n)\log n = \lfloor\frac{1}2 \frac{\log n}{8\lceil\varepsilon(n)\log n\rceil}\rfloor \varepsilon(n)\log n = \Theta(\log n)$. Therefore, we obtain $\liminf\limits_{n \in N} \frac{S(n)T(n)}{n\log n} \ge \liminf\limits_{n \in N} \frac{\varepsilon(n) n \log n \cdot \frac{1}4 k}{n\log n} = \liminf_{n \in N} \frac{\Theta(n\log n)}{n\log n} > 0$, which contradicts the assumption $\lim_{n \in N} \frac{S(n)T(n)}{n\log n}~=~0$.

\emph{Remark.} The assumption $S(n) = \Omega(n)$, which plays its role in the last lines of the proof, is crucial for our construction. The simple information theoretic argument by which we obtained the subfamily $\mathcal{I}$ eliminates any dependency on the $S(n)$ bits of space so that the query algorithm can ``distinguish'' the strings of $\mathcal{I}$ from each other only by probing the input cells. Since each string from the family $\mathcal{I}$ must contain the common ``dictionary'' prefix of length $k\sigma^k$, $k$ cannot exceed $\log n$. The idea of the construction of the sets $T^j_s$ for a string $s \in \mathcal{I}$ is that any algorithm using significantly less than $k$ cell probes (as in the case $T(n) = o(k)$) cannot obtain enough information to distinguish all strings of $\mathcal{I}$. Thus, since $k = O(\log n)$, the best bound for $T(n)$ that one can obtain in this way is $T(n) = \Omega(\log n)$. This is the main reason why it is not immediately clear how to adapt our proof for the case $S(n) = o(n)$.
\end{proof}

\section{Open Problems}

As it follows from the theorem, the data structure of Bille et al.~\cite{BilleEtAl} is optimal when additional space is restricted to $S(n) = \Omega(n)$ bits. Still, there is a $\log n$ gap between the upper and lower bounds for the problem when $S(n) = o(n)$. We believe that the approach of~\cite{BilleEtAl} can be modified to achieve the better $S(n) T(n) = O(n\log S(n))$ trade-off. We conjecture that this trade-off is optimal.

\begin{conjecture}
In the non-uniform cell-probe model the trade-off $S(n)T(n) = \Omega(n\log S(n))$ holds for any data structure that solves the LCE problem for a read-only string of length $n$ using $S(n)$ bits of space and $T(n)$ time for the LCE queries.
\end{conjecture}

There are several other promising directions for further investigations of the LCE data structures.

First, it would be interesting to obtain a version of our result for randomized LCE data structures that answer the LCE queries with high probability.

It is an open problem to provide a tight time and space lower bound on the LCE data structures that use the knowledge of the alphabet size $\sigma$. Further, it is not clear how to generalize our result to the practically important case when $\sigma$ is very small and the algorithm can read $\Omega(\log_{\sigma} n)$ consecutive letters of the input string in $O(1)$ time packing them in one $\Omega(\log n)$-bit machine word.

It seems that still there are many possibilities for improvements of upper bounds for the algorithms that solve the LCE problem. The best currently known algorithm constructing the data structure~\cite{BilleEtAl} is unacceptably slow. Tanimura et al.~\cite{TIBIPT} presented a data structure with significantly faster construction time but with slightly suboptimal query time. Thus, the development of a new optimal LCE data structure with optimal construction time remains an open problem.

For applications that do not consider the input string as read-only, it is an open problem to develop a data structure that provides access to the string in optimal time as in the data structure of Prezza~\cite{Prezza} discussed above and supports the LCE queries in $o(\log\ell)$ time, where $\ell$ is the result of the query. It is also interesting to obtain, if possible, any non-trivial lower bounds for this model.

In practice, randomized construction algorithms for LCE data structures usually behave better than deterministic ones. There is a room for improvements in this direction in the algorithms presented in~\cite{BilleEtAl}, for read-only inputs, and~\cite{Prezza}, for editable inputs. It is also interesting to consider for the later setting the development of more time or space efficient Monte Carlo LCE solutions (for read-only inputs, this problem is exhaustively solved in~\cite{GawrychowskiKociumaka}).

Finally, on the purely theoretical side there are some related open problems in the line of research on general ordered and unordered alphabets. In this classical setting even an LCE data structure with optimal construction time and $O(1)$ query time still was not described. However, there is a strong evidence~\cite{GawrychowskiKociumakaRytterWalen} that, surprisingly, such data structure exists; see~\cite{GawrychowskiKociumakaRytterWalen,Kosolobov}.

\subparagraph{Acknowledgement.}
The author wishes to acknowledge anonymous referees for detailed comments that helped to greatly improve the paper.


\bibliography{lce}
\bibliographystyle{elsarticle-num-names-sort}

\end{document}